\documentclass[
reprint,
amsmath,amssymb,
aps
]{revtex4-1}
\usepackage{graphicx}
\usepackage{dcolumn}
\usepackage{bm}
\usepackage{dsfont}
\usepackage{amsthm}
\usepackage[capitalise]{cleveref}
\crefname{figure}{Fig.}{Fig.}
\newtheorem{theorem}{Theorem}

\newtheorem{lemma}{Lemma}[theorem]
\theoremstyle{definition}
\newtheorem{definition}{Definition}
\theoremstyle{plain}
\newtheorem{defcorollary}{Corollary}[definition]
\theoremstyle{definition}

\begin{document}
\title{Contextuality Test of the Nonclassicality of Variational Quantum Eigensolvers}
\author{William M. Kirby}
\affiliation{Department of Physics and Astronomy, Tufts University\\574 Boston Avenue, Medford, MA 02155}%
\author{Peter J. Love}
\email{peter.love@tufts.edu}
\affiliation{Department of Physics and Astronomy, Tufts University\\574 Boston Avenue, Medford, MA 02155}%
\date{\today}

\begin{abstract}
Contextuality is an indicator of non-classicality, and a resource for various quantum procedures. In this paper, we use contextuality to evaluate the variational quantum eigensolver (VQE), one of the most promising tools for near-term quantum simulation. We present an efficiently computable test to determine whether or not the objective function for a VQE procedure is contextual. We apply this test to evaluate the contextuality of experimental implementations of VQE, and determine that several, but not all, fail this test of quantumness.
\end{abstract}
\maketitle

\section{Introduction}
\label{intro}
Quantum computing hardware is entering the era of noisy intermediate scale quantum (NISQ) computers~\cite{preskill2018quantum}. These are machines that are too large to simulate with classical computers, but too small to allow fault tolerant quantum computation. A crucial question is whether NISQ machines can perform useful tasks beyond the capabilities of classical computers~\cite{NAP25196}.

In the last decade much attention has been focused on algorithms for quantum simulation of chemical systems~\cite{aspuru-guzik05a,whitfield11a,kassal11a,jones12a,yung14a,aspuru-guzik18a,DuH2NMR,lanyon2010towards,peruzzo2014variational,wang2015quantum,omalley16a,santagati18a,shen2017quantum,paesani2017,kandala2017hardware,PhysRevX.8.031022,colless18a,nam19a,kandala19a}.
One such algorithm, the variational quantum eigensolver (VQE, first proposed in \cite{peruzzo2014variational}), has emerged as an important potential application of NISQ computers.
Experimental realizations of VQE have been performed on a number of platforms~\cite{DuH2NMR,lanyon2010towards,peruzzo2014variational,wang2015quantum,omalley16a,santagati18a,shen2017quantum,paesani2017,kandala2017hardware,PhysRevX.8.031022,dumitrescu18a,colless18a,nam19a,kokail19a,kandala19a}.

VQE is based on mapping a Hamiltonian $H$ to a weighted sum $\sum_ih_i\mathcal{P}_i$, where the terms $\mathcal{S}\equiv\{\mathcal{P}_i\}$ are Pauli operators and the $h_i$ are (real) coefficients. A short quantum circuit prepares an ansatz state, and the expectation value of each Hamiltonian term is estimated by repeated prepare-and-measure experiments. The ansatz parameters are optimized classically, producing a variational upper bound to the ground state energy.

VQE is advantageous for NISQ computers because of the short coherence times required compared to phase estimation~\cite{omalley16a}.
Theoretical improvements of VQE to date have proposed methods to reduce the number of qubits and measurements required \cite{love12,hastings15a,tranter15a,low19a,babbush16a,bravyi17a,setia18a,poulin18a,low18a,babbush18c,motta18a,steudtner18a,berry19a}, and to improve the ansatz states \cite{poulin18a,berry18a,tubman18a}, computation of gradients \cite{schuld18a,bergholm18a,napp19a}, and classical optimization techniques \cite{yang17a}. In the present paper we consider a separate issue: how quantum mechanical is this hybrid quantum-classical algorithm, for a given Hamiltonian? We use contextuality as our measure of quantumness.

The study of contextuality began with the Bell-Kochen-Specker theorem \cite{bell64a,bell66a,kochen67a}.
Contextuality of preparation, transformation and measurement were defined in 2008, and the relationship of contextuality to negativity of quasi-probability representations was established~\cite{spekkens08,ferrie08a,ferrie09a,ferrie11a,veitch12a}.  Contextuality has been extensively studied in the last decade~\cite{abramsky11a,ramanathan12a,raussendorf13a,howard14a,cabello14a,cabello15a,ramanathan14a,grudka14a,raussendorf17a,abramsky17a,de_silva17a,amaral17a,horodecki18a,karanjai18a,cabello18a,raussendorf18a,schmid18a,duarte18a,mansfield18a,okay18a,frembs18a,arvidsson-shukur19a,kirchmair09a,leupold18a,raussendorf19b}.

The Bell-Kochen-Specker theorem states that there exist quantum systems for which it is impossible to reproduce the outcome probabilities of every possible measurement as marginals of single joint probability distribution \cite{bell64a,bell66a,kochen67a}.
However, if we restrict to some smaller set of measurements corresponding to a set of observables $\mathcal{S}$, properties of the set determine whether a joint distribution may exist for only those measurements.
Measurement contextuality refers to various types of contradictions that can appear in attempts to describe sets of measurements by joint probability distributions.
We examine ``strong contextuality"~\footnote{so called in \cite{abramsky11a} and studied in ~\cite{abramsky11a,ramanathan12a,abramsky17a,amaral17a,horodecki18a,de_silva17a,karanjai18a,cabello18a,raussendorf18a,schmid18a,duarte18a,raussendorf19b}}, which is contextuality in the same vein as the Peres-Mermin square \cite{peres91a,mermin90a,mermin93a} (see Mermin's outline of a ``plausible" hidden-variable theory in \cite[\S II]{mermin93a}.)
Colloquially, a set of measurements is strongly contextual if it is impossible to consistently assign outcomes to every measurement in the set.
In ``weak" versions of contextuality such as Bell inequality violations, joint outcomes may be consistently assignable, but statistical predictions based on the existence of joint probability distributions are violated.

Since VQE is an important near-term application of NISQ machines, it is natural to consider how the contextuality of VQE procedures is related to any quantum advantage that they may obtain. In this paper, we present a method to analyze the contextuality of VQE procedures. As applied to VQE, strong contextuality is a property of the target Hamiltonian. It is independent of the ansatz states, and provides a stringent test of the quantumness of the problem being addressed. The set of Hamiltonians that are noncontextual by our definition includes diagonal Hamiltonians that encode a classical objective function. Such problems are addressed by the Quantum Approximate Optimization Algorithm (QAOA), which is closely related to VQE~\cite{farhi2014quantum}. As we shall see, the set of noncontextual Hamiltonians contains the set of commuting Pauli Hamiltonians, and therefore represents a broader definition of classicality.

One concept upon which we rely is the \emph{closed subtheory}: a set of measurements in which all measurements whose outcomes are determined with certainty by the outcomes of others in the set are themselves members of the set. We introduce this concept here because it provides a distinction between this work and the criteria for strong contextuality studied in~\cite{ramanathan12a}, which are based on sets of observables that are not necessarily closed subtheories. In~\cite{karanjai18a} it is shown that the efficiency of classical simulation is limited by contextuality for sets of measurements that are closed subtheories. We impose the requirement that sets of operators form closed subtheories, so that the results of~\cite{karanjai18a} apply to our setting.

In~\cite{cabello14a} the authors obtain criteria for contextuality based on compatibility graphs, as do we. However, \cite{cabello14a} focuses on weak contextuality, that is, violation of noncontextual inequalities, whereas our interest is in strong contextuality. We further discuss the distinction between our condition for contextuality and previously studied criteria in \cref{discussion}, and in \cref{relations}.

A natural next step is to develop measures that quantify contextuality based on our criterion. We suggest two simple measures at the end of \cref{contextuality}, and discuss more general measures in \cref{measures}, as well as their relations with prior measures, which include the contextual fraction \cite{abramsky11a,abramsky17a,mansfield18a,duarte18a}, relative entropy of contextuality, mutual information of contextuality, contextual cost (all in \cite{grudka14a}), and rank of contextuality \cite{horodecki18a}.

In \cref{contextuality}, we develop the notion of contextuality we will study and give our main results. In \cref{vqecontextualitytodate} we evaluate the contextuality of several VQE experiments. We conclude in \cref{discussion} with a discussion of our results, and directions for future work. 

\section{Strong contextuality}
\label{contextuality}

We focus on the analysis of strong contextuality for sets of Pauli operators.
We use the following notation:
$X\equiv\sigma_x$, $Y\equiv\sigma_y$, $Z\equiv\sigma_z$, and $I\equiv2\times2$ identity ($\mathds{1}$ will denote a generic identity matrix).
We omit the tensor product symbol: $IX$ denotes $I\otimes\sigma_x$. 
Let $\mathcal{S}$ be the set of measurements that are performed in a VQE procedure: in our case these will be Pauli measurements.
As we will discuss below, the (non)contextuality of a VQE procedure is determined by properties of $\mathcal{S}$.

A \emph{joint outcome assignment} is an assignment of one outcome ($\pm1$) to each measurement in $\mathcal{S}$. In an ontological hidden-variable theory, joint outcome assignments correspond to \emph{ontic states} (``real states") of a system, since they may be interpreted as definite ontological values for the observables $\mathcal{S}$.
A measurement is then seen as revealing information about the ontic state, which exists independently whether it is measured or not.

A \emph{context} on a finite dimensional Hilbert space is a set of pairwise-commuting observables whose eigenvalues uniquely specify the (shared) basis states.
If $\mathcal{S}$ is a context, we will see that it is always possible to consistently assign outcomes to the measurements in $\mathcal{S}$.
However, if $\mathcal{S}$ is not a context and has nonempty intersection with multiple incompatible contexts (context compatibility is defined in \cref{measurementcontextuality}), it may be impossible to consistently assign joint outcomes.
In this case the outcomes thus assigned to any individual measurement are context-dependent: hence the term ``contextual."

Given any set of measurements $\mathcal{S}$, let $\overline{\mathcal{S}}$ be the set of measurements whose outcomes are predicted with certainty given an assignment of outcomes to $\mathcal{S}$. In the language of \cite{karanjai18a}, $\overline{\mathcal{S}}$ corresponds to the smallest closed subtheory containing $\mathcal{S}$. The outcomes for $\overline{\mathcal{S}}$ induced by an assignment of outcomes to $\mathcal{S}$ may contain contradictions even if the outcomes for $\mathcal{S}$ alone do not.

A prediction with certainty occurs when for some observable $A'$ there exists a commuting subset $\mathcal{S}'\subseteq\mathcal{S}$ such that $A'$ is equal to the product of the operators $\mathcal{S}'$. Then since the operators $\mathcal{S}'$ may all be measured simultaneously, in any joint outcome assignment to $\mathcal{S}\cup\{A'\}$ the outcome assigned to $A'$ must be the product of the outcomes assigned to $\mathcal{S}'$: we therefore say that $A'$ is directly determined by $\mathcal{S}$~\cite{peres91a}. $A'$ may now contribute to determining some other operators that are not directly determined by $\mathcal{S}$.
Thus in general a measurement $A$ is determined by $\mathcal{S}$ if there is a ``determining tree" that leads from $\mathcal{S}$ to $A$:
\begin{definition}
    \label{dettreedef}
    A \emph{determining tree} for a Pauli measurement $A$ over a set of Pauli measurements $\mathcal{S}$ is a tree whose nodes are Pauli operators and whose leaves are operators in $\mathcal{S}$, such that...
    \begin{enumerate}
        \item The root is $A$.
        \item All children of any particular parent pairwise commute (as operators).
        \item Every parent node is the operator product of its children (and thus commutes with them).
    \end{enumerate}
\end{definition}

\cref{determiningtreeexample} shows determining trees for the measurements $\pm YY$ over $\mathcal{S}=\{XI,IX,ZI,IZ\}$.
It is easy to check that these trees satisfy the properties of \cref{dettreedef}.
This example is a recasting of the classic Peres-Mermin square \cite{peres91a,mermin90a,mermin93a}.

\begin{figure}[ht]
	\caption{Determining trees for $\pm YY$ over $\{XI,IX,ZI,IZ\}$\label{determiningtreeexample}}
	\centering
	\includegraphics[width=1.2in]{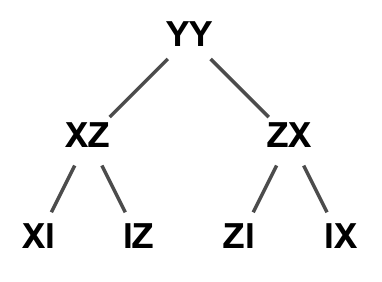}
	\includegraphics[width=1.2in]{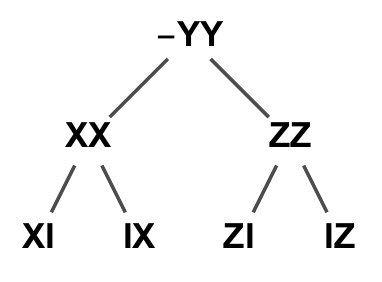}
\end{figure}

Given \cref{dettreedef}, we say that $A$ is \emph{determined} by $\mathcal{S}$ if and only if there exists a determining tree for $A$ over $\mathcal{S}$.
This also provides a formal definition for $\overline{\mathcal{S}}$: it is the set of Pauli measurements for which there exist determining trees over $\mathcal{S}$.

Given a determining tree $\tau$ for a Pauli $A$ over a set of Pauli operators $\mathcal{S}$, and a joint outcome assignment to $\mathcal{S}$, we may now find the determined outcome for $A$. Let $\mathcal{L}$ be the leaves of $\tau$; $\mathcal{L}$ may contain multiple copies of the same operator. By induction on property 3 of a determining tree (see \cref{dettreedef}), $A$ is the operator product of the elements of $\mathcal{L}$.
Therefore, given an assignment of values $\rho_L=\pm1$ to each $L\in\mathcal{L}$, the value assigned to $A$ must be
\begin{equation}
    \label{determiningproduct}
    \rho_A=\prod_{L\in\mathcal{L}}\left(\rho_L\right)^{m_L}=\prod_{L\in\mathcal{D}}\rho_L,
\end{equation}
where the exponent $m_L$ is the multiplicity of the operator $L$ in $\mathcal{L}$.
$\mathcal{D}$ is a subset of the leaves that we call the {\em determining set} of $\tau$, defined as follows:
\begin{definition}
    \label{detsetdef}
    For a determining tree $\tau$, the \emph{determining set} is defined to be the set containing one copy of each operator with odd multiplicity as a leaf in $\tau$.
    If for some determining tree with root $A$, the determining set is empty, then every $m_L$ in the first product in \eqref{determiningproduct} must be even, so the outcome assigned to $A$ is 1.
\end{definition}

We may now state our condition for contextuality:
\begin{definition}
    \label{contextualitydef}
    A set $\mathcal{S}$ of Pauli operators is \emph{contextual} if for some Pauli $A$ there exists a determining tree $\tau$ for $A$ over $\mathcal{S}$ and a determining tree $\tau'$ for $-A$ over $\mathcal{S}$ such that the determining sets for $\tau$ and $\tau'$ are identical.
\end{definition}
By \eqref{determiningproduct}, the existence of such trees implies that for any joint outcome assignment, the outcome for $A$ is both $+1$ and $-1$, which is a contradiction.

How does this apply to the Peres-Mermin square? \cref{determiningtreeexample} gives determining trees for $\pm YY$ over $\mathcal{S}=\{XI,IX,ZI,IZ\}$.
In each tree, the set of leaves is $\mathcal{S}$ and each leaf has multiplicity 1, so the determining set for each tree is $\mathcal{S}$.
Thus $\mathcal{S}$ satisfies the criteria in \cref{contextualitydef}, and is contextual.

The criterion for strong contextuality in \cref{contextualitydef} depends on a measurement operator ($A\in\overline{\mathcal{S}}$) that may or may not be an element of $\mathcal{S}$.
However, for any $\mathcal{S}$ that is contextual according to \cref{contextualitydef}, we may obtain a contradiction in the assignment(s) to an operator contained in $\mathcal{S}$.
This is demonstrated by the following corollary:
\begin{defcorollary}
    \label{contextualitydef2}
    A set $\mathcal{S}$ of Pauli operators is contextual if and only if for some $B\in\mathcal{S}$ there exists a determining tree for $-B$ over $\mathcal{S}$, whose determining set is $\{B\}$.
\end{defcorollary}
The plain language statement of the contradiction in this case is: ``the outcome ($\pm1$) assigned to $-B$ must be the outcome assigned to $B$."
A third equivalent definition is also useful:
\begin{defcorollary}
    \label{contextualitydef3}
    A set $\mathcal{S}$ of Pauli operators is contextual if and only if there exists a determining tree for $-\mathds{1}$ over $\mathcal{S}$, whose determining set is empty.
\end{defcorollary}
\noindent The proofs may be found in \cref{proofs}.
The plain language statement of the contradiction in this case is: ``the outcome assigned to $-\mathds{1}$ (whose eigenvalues are all $-1$) must be $+1$."
\cref{contextualitydef}, \cref{contextualitydef2}, and \cref{contextualitydef3} formalize the notion of contradiction in induced joint outcomes for $\overline{\mathcal{S}}$. Since $\overline{\mathcal{S}}$ is the smallest closed subtheory containing $\mathcal{S}$, such a contradiction constitutes strong contextuality of $\mathcal{S}$.

We now present three theorems that give necessary and sufficient conditions for measurement contextuality in the sense of \cref{contextualitydef}.
We will make use of the following concept:
\begin{definition}
	For a set $\mathcal{S}$ of Pauli operators, the \emph{compatibility graph} of $\mathcal{S}$ is an undirected graph whose nodes are the operators in $\mathcal{S}$, and in which a pair of operators is adjacent if and only if they commute.
\end{definition}

\begin{theorem}
	\label{4formsthm}
	A set of four Pauli operators $\{A,B,C,D\}$ is contextual if and only if its compatibility graph has one of the forms given in \cref{4forms} (up to permutations of the operators).
\end{theorem}

\begin{figure}[h!]
	\caption{Compatibility graphs for contextual sets of four Pauli operators.\label{4forms}}
	\centering
	\includegraphics[width=2.2in]{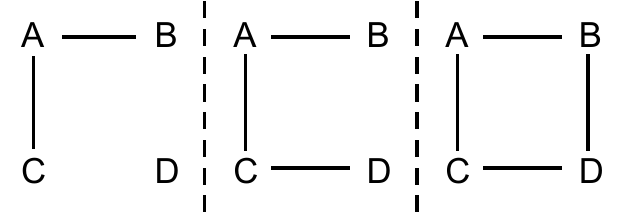}
\end{figure}

\begin{theorem}
	\label{contextualcondition}
	A set of $n$ Pauli operators is contextual if and only if it contains a subset consisting of four operators whose compatibility graph has one of the forms given in \cref{4forms} (up to permutations of the operators).
\end{theorem}

The proofs of \cref{contextualcondition,4formsthm} are given in \cref{proofs}.
\cref{contextualcondition} provides an efficient algorithm for determining whether an arbitrary set $\mathcal{S}$ of Pauli measurements is contextual.
First remove any operators from $\mathcal{S}$ that commute with all others (searching for these takes $O(|\mathcal{S}|^2)$ steps): let $\mathcal{T}$ be the remaining set.
Then, search in $\mathcal{T}$ for a set of three operators $A,B,C$ such that $A$ commutes with $B$ and $C$, but $B$ and $C$ anticommute.
If such a set exists, then since there is some $D\in\mathcal{T}$ that anticommutes with $A$, the compatibility graph of $A,B,C,D$ has one of the forms \cref{4forms} (up to exchange of $B$ and $C$): thus $\mathcal{S}$ is contextual.
If no such set exists, then $\mathcal{S}$ is noncontextual.
There are $O(|\mathcal{S}|^3)$ subsets of size three in $\mathcal{S}$, so this is the runtime for the search.
In many VQE procedures some structure on the set $\mathcal{S}$ is known, which may improve the efficiency of determining whether it is contextual.

Although we ultimately only need to search for triples of operators in the algorithm, the contextual compatibility graphs in \cref{4forms} have four nodes instead of three because we must first remove universally-commuting operators.
Note that after this is done (to obtain $\mathcal{T}$), we search for a subset $\{A,B,C\}$ in which commutation is not transitive.
Each such subset represents an obstacle to commutation being an equivalence relation on $\mathcal{T}$.
This is formalized in the following theorem:

\begin{theorem}
    \label{commutationthm}
    For a set $\mathcal{S}$ of Pauli operators, let $\mathcal{T}\subseteq\mathcal{S}$ be the set obtained by removing any operator that commutes with all others in $\mathcal{S}$.
    Then $\mathcal{S}$ is noncontextual if and only if commutation is an equivalence relation on $\mathcal{T}$.
\end{theorem}
\noindent The proof of~\cref{commutationthm} is given in \cref{proofs}.
That commutation is not transitive in general is a non-classical property.
Operators that commute with all others in the set cannot contribute to contextuality (see Lemma 2.1, in \cref{proofs}), so it is satisfying that after removing these non-transitivity of commutation is equivalent to contextuality.

Can we extend our evaluation procedure to a measure of the amount of contextuality present in a contextual set $\mathcal{S}$? One natural measure of the contextuality of $H$ is obtained by evaluating the distance from $H$ to any noncontextual Hermitian operator, as suggested in \cite{duarte18a}. Any choice of metric on observables will induce such a  measure. Let a \emph{decontextualizing set} $\mathcal{S}'$ be any subset of $\mathcal{S}$ such that $\mathcal{S}\setminus\mathcal{S}'$ is noncontextual.
Then we may define another measure of contextuality as the minimum of $\sum_j|h'_j|$ over all subsets $\{h'_j\}$ of the coefficients that are associated to decontextualizing sets. This measure provides an upper bound on the error in the energy estimate induced by ``decontextualizing" the Hamiltonian. We discuss generalizations of these measures, and their relations with previously studied measures in \cref{measures}.

\section{Evaluation of Contextuality in VQE experiments to date}
\label{vqecontextualitytodate}

We now use the methods in \cref{contextuality} to assess contextuality in VQE experiments performed to date.
The results are summarized in \cref{vqecontextualitytable}, in which we also give $\text{CD}_0$, a measure of contextuality given by the minimum size of any decontextualizing set as a fraction of the total number of terms.
For the larger Hamiltonians, we use a heuristic approximation for $\text{CD}_0$: see \cref{measures} for details about this method and about the experiments.
Note that each simulation of $H_2$ in the STO-3G minimal basis is noncontextual.
This is not surprising if one considers these simulations as encoding a two-dimensional Hilbert space spanned by a bonding and antibonding state, i.e., a single qubit, for which Bell gave a noncontextual hidden-variable theory~\cite{bell66a}.

\begin{table}[htp]
  \begin{tabular}{ l l c r r}
    Citation: & System: & \hspace{-0.15in}Contextual? & $\text{CD}_0$ & $|\mathcal{S}|$\\
    \hline
    Dumitrescu \emph{et al.} \cite{dumitrescu18a}\ \ & Deuteron\quad\ & No & 0 & --- \\
    \hline
    Kandala \emph{et al.} \cite{kandala2017hardware}\quad\ & H$_2$\quad\ & No & 0 & 4 \\
    \hline
    O'Malley \emph{et al.} \cite{omalley16a}\quad\ & H$_2$\quad\ & No & 0 & 5 \\
    \hline
    Hempel \emph{et al.} \cite{PhysRevX.8.031022}\quad\ & H$_2$ (BK)\quad\ & No & 0 & 5 \\
    \hline
    Hempel \emph{et al.} \cite{PhysRevX.8.031022}\quad\ & H$_2$ (JW)\quad\ & No & 0 & 14 \\
    \hline
    Colless \emph{et al.} \cite{colless18a}\quad\ & H$_2$\quad\ & No & 0 & 5 \\
    \hline
    Kokail \emph{et al.} \cite{kokail19a}\quad\ & \hspace{-0.3in}Schwinger Model & Yes & $\sim$0.14 & ~231 \\
    \hline
    Nam \emph{et al.} \cite{nam19a}\quad\ &$\text{H}_2$O\quad\ & Yes & 0.27 & 22 \\
    \hline
    Hempel \emph{et al.} \cite{PhysRevX.8.031022}\quad\ & LiH\quad\ & Yes & 0.33 & 13 \\
    \hline
    Peruzzo \emph{et al.} \cite{peruzzo2014variational}\quad\ & HeH$^+$\quad\ & Yes & 0.38 & 8 \\
    \hline
    Kandala \emph{et al.} \cite{kandala2017hardware}\quad\ & BeH\quad\ & Yes & $\sim$0.74 & 164 \\
    \hline
    Kandala \emph{et al.} \cite{kandala2017hardware,kandala19a}\quad\ & LiH\quad\ & Yes & $\sim$0.77 & 99
  \end{tabular}
\caption{Evaluation of contextuality in VQE experiments. $\text{CD}_0$ is the minimum number of terms we must remove from the Hamiltonian to reach a noncontextual set, as a fraction of the total number of terms ($|\mathcal{S}|$). In \cite{dumitrescu18a}, $|\mathcal{S}|$ varies.}\label{vqecontextualitytable}
\end{table}

\section{Discussion}
\label{discussion}

All VQE procedures that have been implemented to date, whether noncontextual or contextual, have been small enough to simulate classically.
The purpose of such experiments is not to demonstrate quantum advantage, but to apply current hardware to small examples of real-world applications.
Such efforts have been instrumental in developing both experimental and theoretical capabilities; indeed, VQE itself was developed in this context~\cite{peruzzo2014variational}.

For these reasons, we should be clear that our classification of these experiments as contextual or noncontextual is not a judgement of the value of the experiments, but rather a constructive categorization whose purpose is to inform future experiments and theoretical work.
Contextuality of a Hamiltonian according to our definition is connected to inefficiency of classical simulation~\cite{karanjai18a}.
Furthermore, as noted above, we may regard a noncontextual Hamiltonian as an instance of an essentially classical problem, akin to quantum algorithms for explicitly classical problems as in QAOA~\cite{farhi2014quantum} (note that QAOA's diagonal Hamiltonians are always noncontextual.)

In spite of this last point, however, a noncontextual VQE procedure may still be hard to simulate classically, since classical problems can be classically hard.
However, contextuality in a VQE procedure provides a strict separation between it and any classical algorithm, by ruling out the existence of a description of the problem in terms of joint probability distributions over a classical phase space, and thus precluding any classical approach either explicitly or implicitly based on such distributions.
We suggest therefore that future VQE implementations, even at small scales, should focus on contextual Hamiltonians, according to the criteria we have developed.

Our criterion for contextuality of a set of Pauli operators $\mathcal{S}$ is that joint outcome assignments to $\overline{\mathcal{S}}$ are necessarily self-contradictory.
In other words, we analyze contextuality for the minimal closed subtheory containing $\mathcal{S}$; this allows us to invoke the results of \cite{karanjai18a}, which show that efficient simulation by sampling from the discrete Wigner function is only possible in the absence of contextuality.
This is not the only choice: for example, \cite{abramsky11a,ramanathan12a,cabello18a} do not require the measurements to form a closed subtheory.
The relationship of our criterion to that of \cite{abramsky11a,ramanathan12a,cabello18a} is discussed further in \cref{relations}.

The set of noncontextual Hamiltonians contains the set of commuting Pauli Hamiltonians, but is distinct from the set of frustration-free Hamiltonians, as may be seen by decomposing two consecutive projectors in the AKLT model (e.g., \cite{affleck2004rigorous}) into Pauli operators.
We leave further consideration of the set of noncontextual Hamiltonians to future work.

Subsequent to the appearance of our work, the result given in our \cref{contextualcondition} was independently discovered in \cite[\S IV]{raussendorf19b}, which presents a Wigner function treatment of qubit systems using a phase space constructed from noncontextual closed subtheories.

\section*{Acknowledgements}
W.M.K. acknowledges support from the National Science Foundation, Grant No. DGE-1842474. P.J.L. acknowledges support from the National Science Foundation, Grant No. PHY-1720395, and from Google, Inc. This work was supported by the NSF STAQ project (PHY-1818914).

\bibliography{references}

\appendix

\section{Proofs}
\label[appendix]{proofs}

We first show that we may restrict our attention to binary determining trees (see Definition 1, in the main text).
This lemma does not appear in the main text, but will be useful in the proofs that follow.
\setcounter{theorem}{1}
\begin{lemma}
\label{binarydettree}
    Given any $k$-ary determining tree, there exists an equivalent binary determining tree, i.e., one that has the same leaves and root.
\end{lemma}
\begin{proof}
    A $k$-ary determining tree is one in which each parent has at most $k$ children.
    Given any $k$-ary determining tree, consider any parent $A$ with $m$ children $\{B_1,...,B_m\}$, for $2<m\le k$: then
    \begin{equation}
        A=\prod_{i=1}^mB_i.
    \end{equation}
    Since the $B_i$ all commute, $B_m$ commutes with
    \begin{equation}
        C\equiv\prod_{i=1}^{m-1}B_i.
    \end{equation}
    Therefore, we may replace the $\{B_1,...,B_{m-1}\}$ as children of $A$ by $C$, which itself has $\{B_1,...,B_{m-1}\}$ as children.
    $A$ now has only $2$ children, and the new node $C$ has $m-1$ children.
    We iterate this operation to obtain a cascade of parents with exactly $2$ children, terminating with the parent $B_1B_2$, which has children $B_1$ and $B_2$.
    Applying this process to every parent in the original $k$-ary tree results in an equivalent binary determining tree, in the sense that the root and leaves are identical and it is still a valid determining tree.
\end{proof}
\setcounter{theorem}{0}

We now give proofs for the results presented in the main text.

\noindent
\textbf{Corollary 3.1.}
\emph{
    A set $\mathcal{S}$ of Pauli operators is contextual if and only if for some $B\in\mathcal{S}$ there exists a determining tree for $-B$ over $\mathcal{S}$, whose determining set is $\{B\}$ (the set containing the single element $B$).
}
\begin{proof}
    ``If" follows because $B$ can be taken to be a determining tree for itself (it is both the root and the single leaf).
    If there also is a determining tree for $-B$ whose determining set is $\{B\}$, then these two trees satisfy the criteria of Definition 3 in the main text.
    
    ``Only if": given determining trees for $\pm A$ with the same determining set $\mathcal{D}$, we may select any operator $B$ in $\mathcal{D}$ and construct a determining tree for $-B$ whose determining set is exactly $\{B\}$.
    This implies that for any outcome assignment $\rho_B=\pm1$ to $B$, $\rho_B=-\rho_B$, a contradiction.

    The construction goes as follows.
    Let the two determining trees be $\tau$ and $\tau'$, as in Definition 3 in the main text.
    Assume that $\tau$ is in binary tree form (see \cref{binarydettree}).
    Let $p=\{C_0,C_1,C_2,...,C_d\}$ denote the path from $A$ to $B$ in $\tau$, where $d$ is the depth of $\tau$ (so $C_0=A$ and $C_d=B$), and let $C_i'$ be the sibling of $C_i$ for each $i>0$.
    Further, let $\{\tau_1,\tau_2,...,\tau_d\}$ be the subtrees of $\tau$ with $\{C_1',C_2',...,C_d'\}$ as their roots.
    Thus, every leaf of $\tau$ except $B$ itself is a leaf of exactly one of the $\tau_i$.

    First, construct a new determining tree $T_1$ for $-C_1$ by letting its children be $-A$ (with $\tau'$ attached) and $C_1'$ (with $\tau_1$ attached).
    $T_1$ is a determining tree for $-C_1$ because $C_1C_1'=A$ (since this is the topmost parent-children group in $\tau$), and thus $-AC_1'=-C_1$, since $C_1'$ is self-inverse and commutes with $C_1$ and $A$.
    The leaves in $T_1$ are the leaves of $\tau'$ as well as the leaves of $\tau_1$.
    Next, construct a new determining tree $T_2$ for $-C_2$ by letting its children be $-C_1$ (with $T_1$ attached) and $C_2'$ (with $\tau_2$ attached).
    As in the previous step, $C_2C_2'=C_1\ \Rightarrow\ -C_1C_2'=-C_2$, so this is a valid determining tree, and its leaves are the leaves of $\tau'$, $\tau_1$, and $\tau_2$.
    We iterate this step until we obtain a determining tree $T_d$ for $-C_d=-B$.
    The leaves of $T_d$ will be the leaves of $\tau',\tau_1,\tau_2,...,\tau_d$.
    As we noted above, the set of leaves of $\tau_1,\tau_2,...,\tau_d$ is exactly the set of leaves of $\tau$, except for $B$ itself, so the set of leaves of $T_d$ is the set of leaves of $\tau$ plus the set of leaves of $\tau'$, minus exactly one copy of $B$.
    By assumption, $\tau$ and $\tau'$ have the same determining set $\mathcal{D}$.
    Thus, every element of $\mathcal{D}$ appears as a leaf in $T_d$ an even number of times (as does every other leaf in $\tau$ and $\tau'$), except for $B$, which appears an odd number of times.
    Therefore, the determining set for $T_d$ is exactly $\{B\}$, so as discussed above, we have a contradiction.
\end{proof}

We illustrate this construction using the Peres-Mermin square example.
Let $\tau$ and $\tau'$ be the left and right determining trees in Fig. 1 in the main text, and $A=YY$.
Let us choose $B$ to be $XI$.
Then the construction requires two steps.
In the first step, $C_1=XZ$, $C_1'=ZX$: $\tau_1$ is given by the first tree in \cref{contextualitydef2example}.
The resulting $T_1$ is a determining tree for $-C_1=-XZ$, given by the second tree in \cref{contextualitydef2example}.
In the second step, $C_2=B=XI$, $C_2'=IZ$, and $\tau_2$ is simply the node $IZ$, since $IZ$ has no children in $\tau$.
Thus $T_2$ is a determining tree for $-C_2=-XI$, given by the final tree in \cref{contextualitydef2example}.
Notice that each operator appears twice as a leaf except for $XI$ itself, so we have a contradiction, as expected.

\begin{figure}[ht]
	\caption{Example of Corollary 3.1: steps to construct a contextual determining tree for $-XI$, beginning with the Peres-Mermin square trees in Fig. 1 in the main text. Note that $-XI$ is the root of $T_2$, the final tree, and $\{XI\}$ is the determining set, since every other leaf appears twice.\label{contextualitydef2example}}
	\centering
	\includegraphics[width=0.85in]{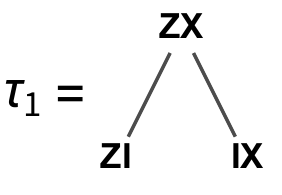}\\
    \includegraphics[width=2in]{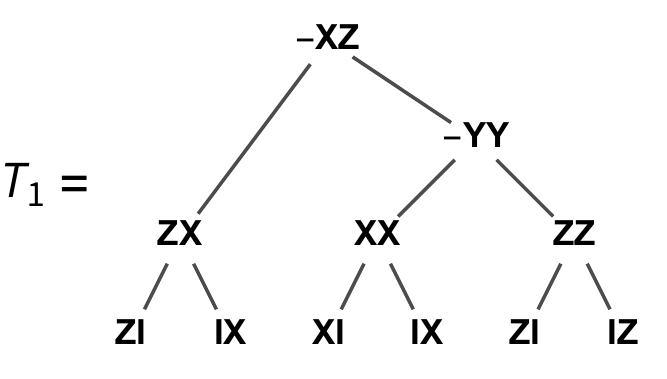}\\
    \includegraphics[width=2.5in]{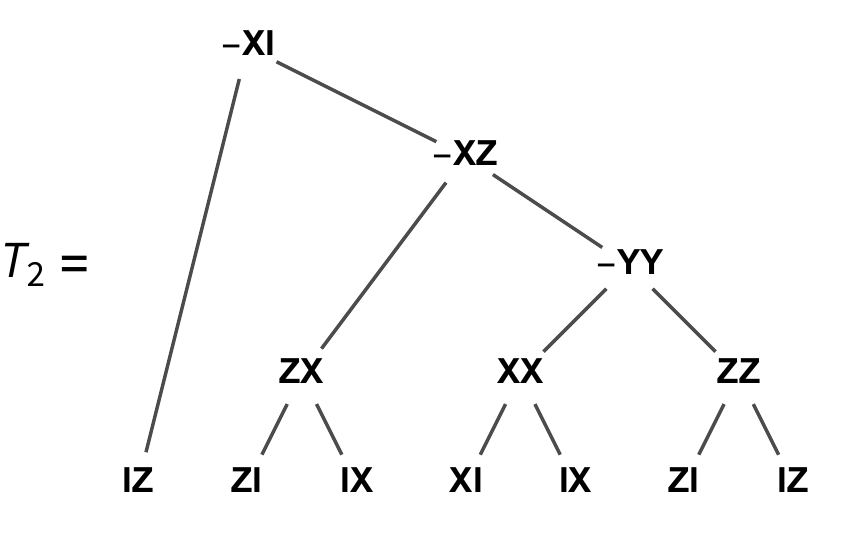}
\end{figure}

\noindent
\textbf{Corollary 3.2.}
\emph{
    A set $\mathcal{S}$ of Pauli operators is contextual if and only if there exists a determining tree for $-\mathds{1}$ (the identity of appropriate dimension) over $\mathcal{S}$, whose determining set is empty.
    (As noted in Definition 2 in the main text, if a determining set for an operator is empty it means that every outcome for that operator is 1, which in this case is an immediate contradiction.)
}
\begin{proof}
    ``Only if": Let $\mathcal{S}$ be contextual.
    Then by Definition 3 in the main text, there exists some Pauli $A$ and determining trees $\tau$ and $\tau'$ for $A$ and $-A$ (respectively) over $\mathcal{S}$, such that the determining sets of $\tau$ and $\tau'$ are identical.
    We construct a determining tree $\tau''$ whose root is $-\mathds{1}$ with children $A$ and $-A$, and $\tau$ and $\tau'$ attached to these.
    Since the determining sets of $\tau$ and $\tau'$ are identical, each element in them appears an even number of times as a leaf in $\tau''$, and is thus not in its determining set.
    Leaves of $\tau$ and $\tau'$ that are not in their determining sets are also not in the determining set of $\tau''$, so the determining set of $\tau''$ must be empty.
    
    ``If": Suppose there exists a determining tree $\tau$ for $-\mathds{1}$ over $\mathcal{S}$ whose determining set is empty.
    By \cref{binarydettree}, we may take $\tau$ to be a binary tree.
    Then the children of $-\mathds{1}$ in $\tau$ must be $A$ and $-A$ for some $A\in\mathcal{S}$.
    Every leaf in $\tau$ is therefore either a leaf in the subtree $\tau'$ whose root is $A$ or a leaf in the subtree $\tau''$ whose root is $-A$.
    Since the determining set of $\tau$ is empty, each leaf $L$ in $\tau$ appears an even number of times.
    Therefore, $L$ either appears an even number of times in both $\tau'$ and $\tau''$ (in which case it is in neither determining set), or an odd number of times in both (in which case it is in both determining sets).
    Since this argument applies to every leaf in $\tau$, the determining sets of $\tau'$ and $\tau''$ must be identical, and thus $\mathcal{S}$ is contextual, by Definition 3 in the main text.
\end{proof}

\begin{figure}[h!]
	\caption{Compatibility graphs for contextual sets of four Pauli operators. (This figure also appears as Fig. 2 in the main text.)\label{4forms_app}}
	\centering
	\includegraphics[width=2.2in]{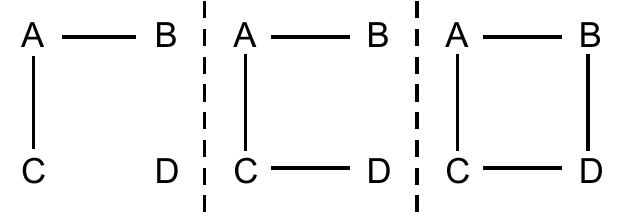}
\end{figure}

\setcounter{theorem}{0}
\begin{theorem}
	\label{4formsthm}
	A set of four Pauli operators $\{A,B,C,D\}$ is contextual if and only if its compatibility graph has one of the forms given in \cref{4forms_app} (up to permutations of the operators).
\end{theorem}
\begin{proof}
	``If": To show that any set of four Pauli operators whose compatibility graph has one of the forms \cref{4forms_app} is contextual, we construct a determining tree for $-\mathds{1}$ with empty determining set over a set of Pauli operators with each form.
	The three compatibility graphs and their corresponding trees are shown in \cref{4formsfigure1,4formsfigure2,4formsfigure3}.
	Using the compatibility graphs, one can easily validate the determining trees, and thus each of the corresponding sets of Pauli operators is contextual.
	
	``Only if" will follow directly from the ``Only if" implication in \cref{contextualcondition}.
\end{proof}
\begin{figure}[ht]
	\caption{First compatibility graph in \cref{4forms_app}, together with its contextual determining tree. Note that each leaf in the tree appears twice, so the determining set is empty.\label{4formsfigure1}}
	\centering
	\includegraphics[width=2.15in]{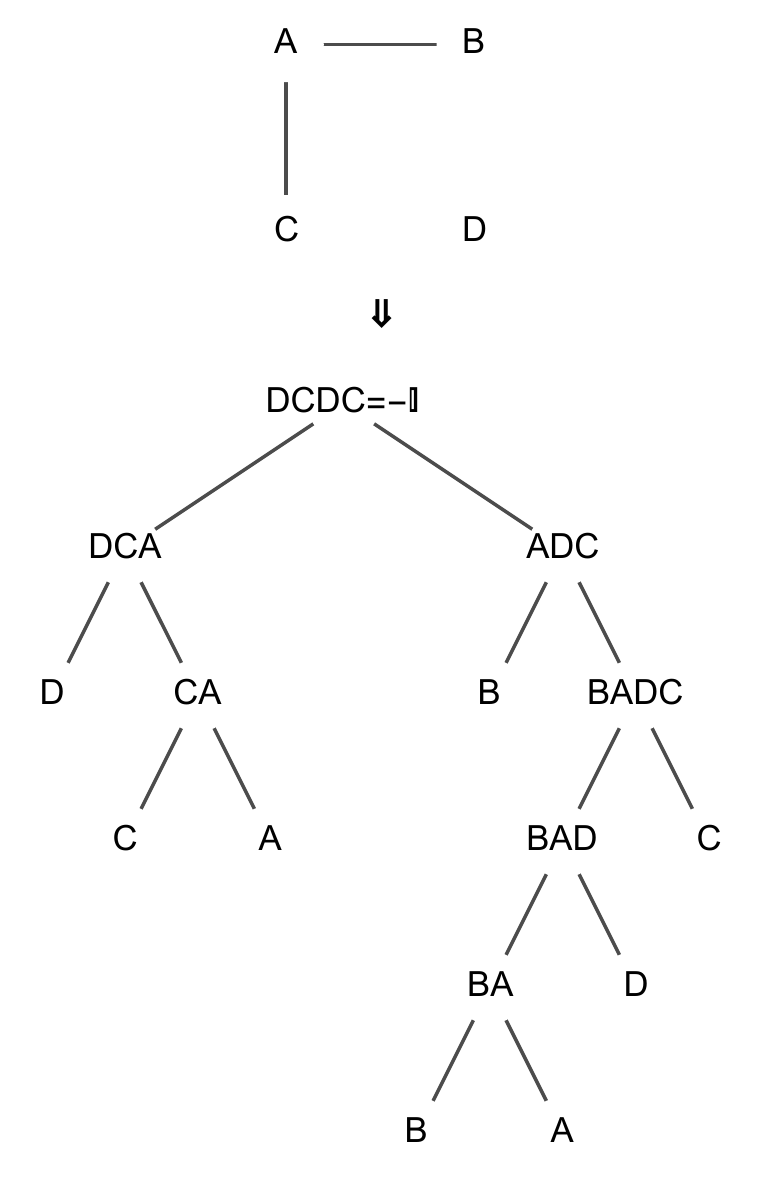}
\end{figure}
\begin{figure}[ht]
	\caption{Second compatibility graph in \cref{4forms_app}, together with its contextual determining tree. Note that each leaf in the tree appears twice, so the determining set is empty.\label{4formsfigure2}}
	\centering
	\includegraphics[width=2.5in]{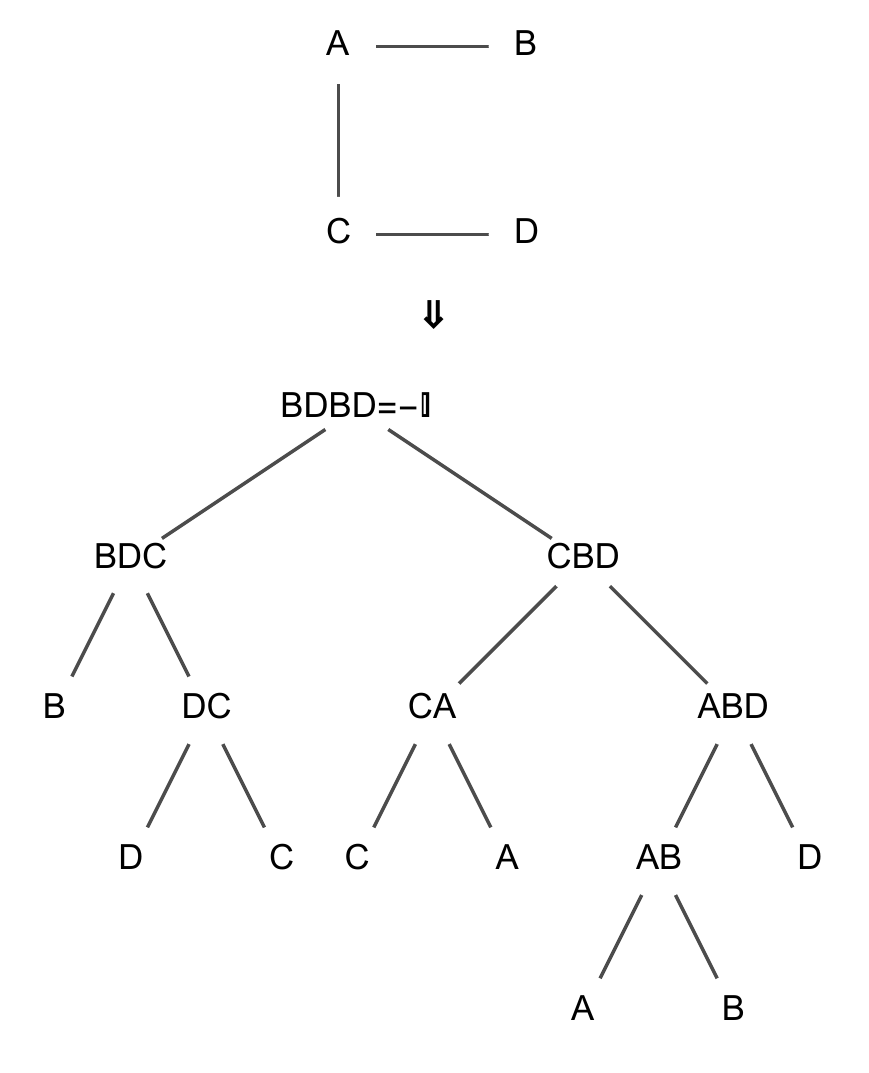}
\end{figure}
\begin{figure}[ht]
	\caption{Third compatibility graph in \cref{4forms_app}, together with its contextual determining tree. Note that each leaf in the tree appears twice, so the determining set is empty.\label{4formsfigure3}}
	\centering
	\includegraphics[width=3.2in]{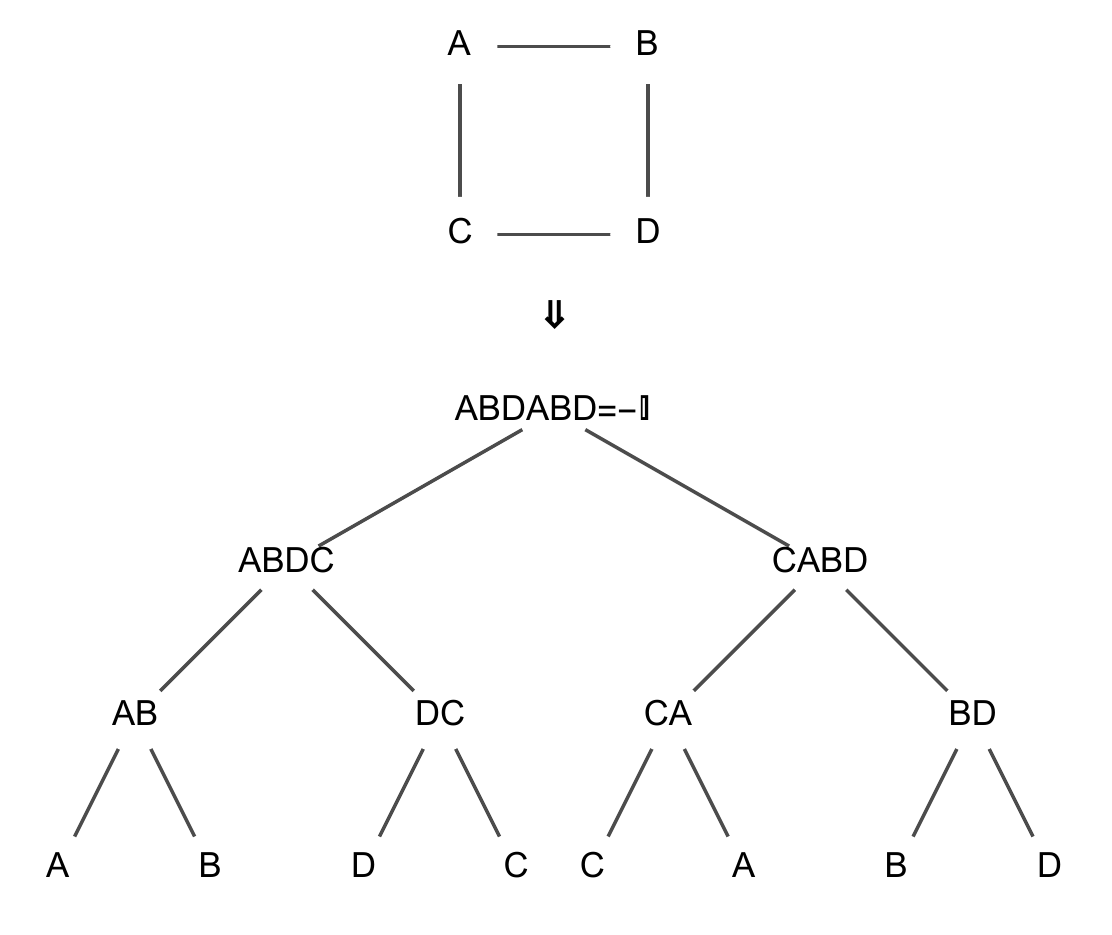}
\end{figure}

We now prove two lemmas that will be useful in the proof of Theorem 2.

\setcounter{theorem}{2}
\begin{lemma}
    \label{lemma1}
    Suppose $\mathcal{T}$ is a set of Pauli operators and $\tau$ is a determining tree over $\mathcal{T}$ with root $R$.
    If any Pauli operator $F$ (which need not be in $\mathcal{T}$) commutes with every operator in $\mathcal{T}$, and $F$ is a node in $\tau$, we may construct a new determining tree $\tau'$ over $\mathcal{T}$ with root $R$, in which $F$ only appears as a child of $R$.
\end{lemma}
\begin{proof}
    Suppose $F$ is a node in $\tau$ that is not a child of $R$.
    Then there is a subtree of $\tau$ with depth 3 (contains 3 layers), containing $F$ as a leaf (assume $\tau$ has been written in binary form).
    This subtree is the left tree in \cref{subprooffigure}: $F$, $\tau_1$, $\tau_2$, and $\tau_3$ may themselves be the roots of subtrees.
	Since $F$ commutes with all operators in $\mathcal{T}$, it commutes with any product of them, so it commutes with $\tau_2\tau_3$.
	Thus since $F\tau_1$ commutes with $\tau_2\tau_3$ (they are children of the same parent), $\tau_1$ must commute with $\tau_2\tau_3$ as well.
	So, we can replace the whole subtree by the right tree in \cref{subprooffigure}.
	Note that $F$ is now a child of the root in this subtree.
	By repeating transformations of this form, we may move all instances of $F$ up until they are children of $R$, the root of $\tau$.
\end{proof}
\begin{figure}[h!]
	\caption{Sequence of trees in \cref{lemma1}, demonstrating that since $F$ commutes universally, it may be moved up to become a child of the root.\label{subprooffigure}}
	\centering
	\includegraphics[width=3.2in]{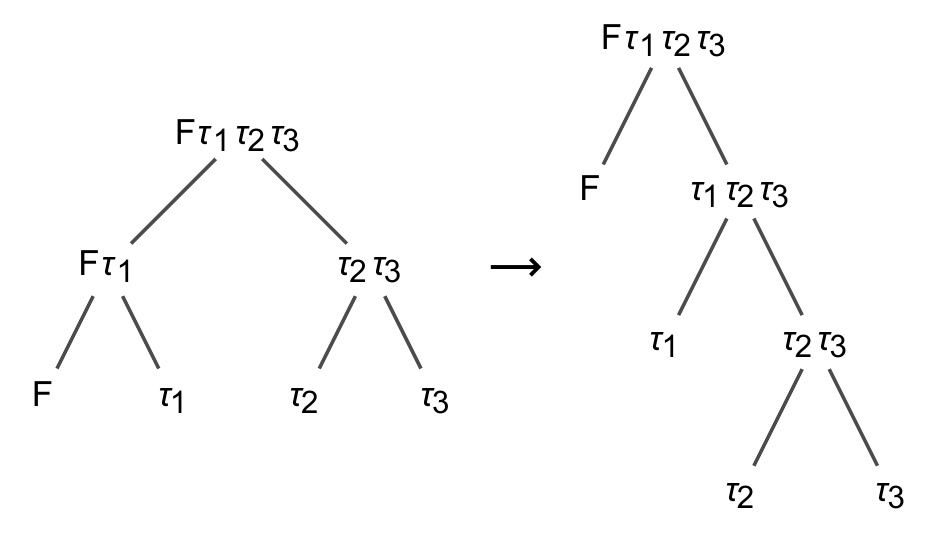}
\end{figure}

\begin{lemma}
    \label{lemma2}
    If the compatibility graph $G$ for a set of Pauli operators $\mathcal{T}$ is a disjoint union of cliques (complete subgraphs), then $\mathcal{T}$ is noncontextual.
\end{lemma}
\begin{proof}
    Suppose $G$ is a disjoint union of cliques.
    Consider any determining tree over $\mathcal{T}$ with empty determining set and root $R$.
    Since we may take the tree to be binary, and only leaves within the same clique in $G$ commute, if two leaves have the same parent then they must be in the same clique.
    For any particular commuting clique $\mathcal{C}$, the parent (product) of two operators in $\mathcal{C}$ commutes with every operator in $\mathcal{T}$, since operators not in $\mathcal{C}$ anticommute with each operator in the product.
    Therefore, we may use \cref{lemma1} to move all such parents up until they are children of $R$.
    The resulting tree is the first tree in \cref{lemma2figure}, where $P_1,P_2,...,P_k$ are parents of pairs of leaves in same clique.
    Since these all commute, we may group them by clique and merge the products within cliques, obtaining the second tree in \cref{lemma2figure}, where $C_1,C_2,...,C_m$ are parents of sets of leaves in the same clique, with one parent per clique (let there be $m$ cliques.)
    The remaining subtrees $\tau_1$ and $\tau_2$ are the remnants of the subtrees whose roots were the original children of $R$.
    Since all parents of pairs of leaves in the same clique have been removed from $\tau_1$ and $\tau_2$, no parent of leaves in $\tau_1$ or $\tau_2$ can be a parent of more than one leaf.
    But a parent of a single leaf is just that leaf, so we lose no generality by removing such leaves.
    Therefore, $\tau_1$ and $\tau_2$ cannot contain any parents of leaves, and therefore must be leaves themselves.
    But since they must commute, they must either be in the same clique, or one or both must be the identity (i.e., not actually present in the tree.)
    
    In either case, we can merge them with the product of leaves from their (shared) clique, resulting in the final tree in \cref{lemma2figure}, where $C_1,C_2,...,C_m$ are parents of sets of leaves in the same clique, with one parent per clique.
    If $\tau_1$ and $\tau_2$ are not present, then the tree already has this form.
    Since each $C_i$ is the product of all the leaves from the same clique, and each leaf must appear an even number of times (since the determining set is empty), every $C_i$ is $\mathds{1}$, and thus the root of the tree is $\mathds{1}$ as well.
	Therefore, since this holds for any determining tree over $\mathcal{T}$ with empty determining set, $\mathcal{T}$ is noncontextual, by Corollary 3.2.
\end{proof}
\begin{figure}[h!]
	\caption{Sequence of trees in \cref{lemma2}. In the first tree, the $P_i$ are parents of pairs of leaves, and $\tau_1,\tau_2$ are subtrees that are remnants of the original children of $R$. In the second tree, the $C_i$ are parents of leaves from within the same commuting clique, with one $C_i$ for each clique. In the third tree, the $\tau_i$ have been absorbed into the $C_i$.\label{lemma2figure}}
	\centering
	\includegraphics[width=1.8in]{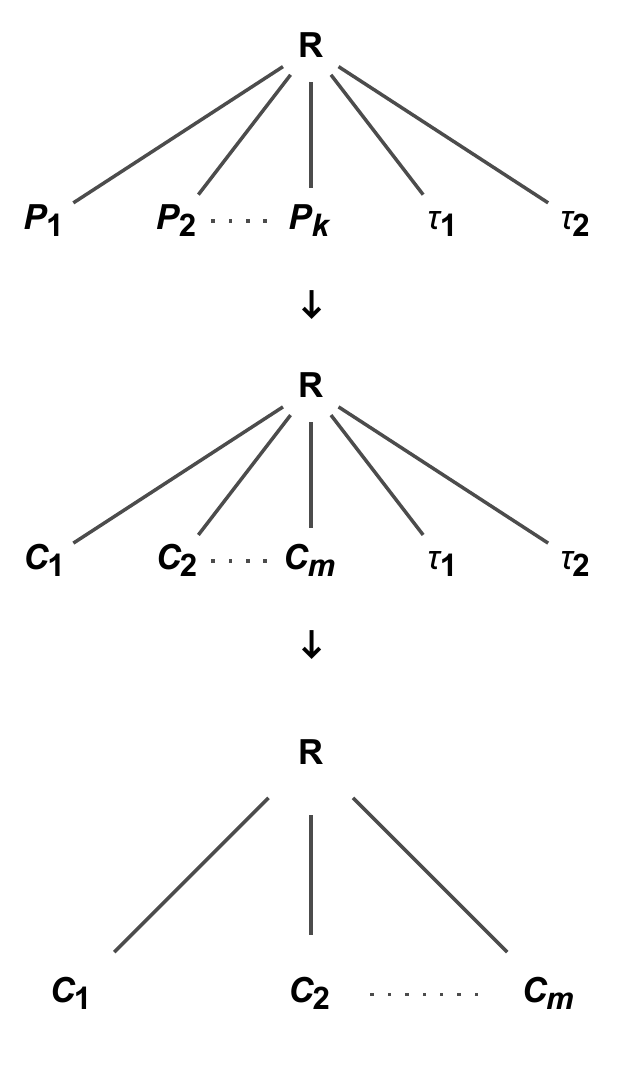}
\end{figure}
Fully-commuting or fully-anticommuting sets of Pauli operators are trivial examples of sets whose compatibility graphs are disjoint unions of cliques.
A less trivial example is the set $\{XX,ZZ,XZ,ZX\}$, in whose compatibility graph $\{XX,ZZ\}$ and $\{XZ,ZX\}$ are disjoint cliques (the reader may recognize these cliques as the intermediate layers in the two determining trees corresponding to the Peres-Mermin square, Fig. 1 in the main text.)
Thus \cref{lemma2} indicates that this set on its own is noncontextual.

\setcounter{theorem}{1}
\begin{theorem}
	\label{contextualcondition}
	A set of $n$ Pauli operators is contextual if and only if it contains a subset consisting of four operators whose compatibility graph has one of the forms given in \cref{4forms_app} (up to permutations of the operators).
\end{theorem}
\begin{proof}
   	``If" follows immediately from the ``If" implication in \cref{4formsthm}, since a set of measurements is contextual if any of its subsets are.
    
	``Only if": Let $\mathcal{S}$ be a contextual set of Pauli operators.
	Then there exists a determining tree $\tau$ for $-\mathds{1}$ over $\mathcal{S}$, with empty determining set, by Corollary 3.2.
	Let $\mathcal{T}\subseteq\mathcal{S}$ be the set of leaves of $\tau$.
	
	By \cref{lemma1}, we may move all instances of any leaf $C\in\mathcal{T}$ that commutes with every other operator in $\mathcal{T}$ up until they are children of the root.
	Since $C$ is self-inverse and must have even multiplicity (since the determining set of $\tau$ is empty), all of its instances now cancel, so we may remove them from the tree.
	Repeating this process for every leaf that commutes with every operator in $\mathcal{T}$ results in a new tree in which every leaf anticommutes with at least one other, so we assume that this is now the case.
	
    	Since $\mathcal{T}$ is contextual, by \cref{lemma2} the compatibility graph of $\mathcal{T}$ cannot be a disjoint union of cliques.
    	Therefore, commutation is not an equivalence relation over $\mathcal{T}$, so there must exist $A,B,C\in\mathcal{T}$ such that $A$ commutes with $B$ and $C$, but $B$ and $C$ anticommute.
	We argued above we may take $\tau$ to be such that every operator in $\mathcal{T}$ anticommutes with at least one other.
	Thus, must exist $D\in\mathcal{T}$ such that $A$ anticommutes with $D$.
	Therefore, the compatibility graph of $\{A,B,C,D\}\subseteq\mathcal{S}$ must have one of the contextual forms given in \cref{4forms_app} (up to swapping $B$ and $C$ in the second form).
\end{proof}

Note that stating that the compatibility graph for $\mathcal{T}$ is a disjoint union of cliques (the conditional in \cref{lemma2}) means that for any $A,B,C\in\mathcal{T}$, if $A$ commutes with both $B$ and $C$, then $B$ and $C$ must also commute.
In other words, it is equivalent to stating that commutation is an equivalence relation when restricted to $\mathcal{T}$.
We formalize this in the following theorem:

\begin{theorem}
	\label{commutationthm}
    	For a set $\mathcal{S}$ of Pauli operators, let $\mathcal{T}$ be the set obtained by removing any operator that commutes with all others in $\mathcal{S}$.
    	Then $\mathcal{S}$ is noncontextual if and only if commutation is an equivalence relation on $\mathcal{T}$.
\end{theorem}
\begin{proof}
    	``If": if there exists any determining tree $\tau$ over $\mathcal{S}$ for $-\mathds{1}$ with empty determining set, then as noted in the "Only if" portion of the proof of \cref{contextualcondition}, by \cref{lemma1} we may obtain an equivalent determining tree over $\mathcal{T}$ (i.e., a determining tree for $-\mathds{1}$ whose determining set is also empty, and whose leaves are in $\mathcal{T}$.)
    	Thus, the (non)contextuality of $\mathcal{S}$ is identical to the (non)contextuality of $\mathcal{T}$.
   	 By \cref{lemma2}, if $\mathcal{T}$ is a disjoint union of cliques, then it is noncontextual, and thus by the above argument, $\mathcal{S}$ is as well.
    
    	``Only if": by the argument that concludes the proof of \cref{contextualcondition}, if commutation is not an equivalence relation on $\mathcal{T}$, then there is a subset of $\mathcal{T}$ whose compatibility graph has one of the forms \cref{4forms_app}.
    	By \cref{4formsthm}, any such subset is contextual, and thus $\mathcal{T}$ is as well.
\end{proof}

\section{Measurement contexts}
\label[appendix]{measurementcontextuality}

Let $\mathcal{H}$ be a finite dimensional Hilbert space.
A \emph{context} for $\mathcal{H}$ is a complete, commuting set of observables: a set of pairwise-commuting observables such that the shared eigenvectors are uniquely specified by their eigenvalue sets under the observables.
Since $\mathcal{H}$ is spanned by the eigenvectors of any observable, the shared eigenvectors of a context form a basis for $\mathcal{H}$.
If $\mathcal{H}$ is the Hilbert space of $n$ qubits, any context composed of two-outcome observables will contain $n$ such observables.

Two observables are said to be \emph{compatible} if they commute.
Compatibility of observables is not transitive, and hence is not an equivalence relation.
For example, given a system of two qubits, $XI$ and $ZI$ both commute with $IX$, but not with each other.
We may generalize the definition of compatibility to apply to contexts: two contexts are compatible if all of their observables pairwise commute.
\begin{theorem}
	Context compatibility is an equivalence relation.
\end{theorem}
\begin{proof}
    Since commutativity of operators is reflexive ($[A,A]=0$) and symmetric ($[A,B]=0\Leftrightarrow[B,A]=0$), commutativity of contexts is as well.
	It remains to demonstrate that commutativity of contexts is transitive.
	
	Let $\sim$ denote context compatibility.
	Let $A=\{A_i\}$, $B=\{B_i\}$, and $C=\{C_i\}$ be contexts for a Hilbert space $\mathcal{H}$ with eigenbases $\{|a\rangle\}$, $\{|b\rangle\}$, and $\{|c\rangle\}$, so that for example,
	$A_i|a\rangle=a_i|a\rangle$
	where $a_i$ is the eigenvalue of $A_i$ labeling $|a\rangle$.
	Suppose $A\sim B$ and $A\sim C$.
	For any pair of commuting observables there exists a shared eigenbasis.
	Therefore since the eigenvectors in each basis $\{|a\rangle\}$, $\{|b\rangle\}$, $\{|c\rangle\}$ are uniquely specified by the eigenvalues of the observables in the associated contexts,
	$\{|a\rangle\}=\{|b\rangle\}$ and $\{|a\rangle\}=\{|c\rangle\}$
	(up to phases).
    Thus $\{|b\rangle\}=\{|c\rangle\}$ up to phases as well, so each is a set of shared eigenvectors of the observables in $B$ and $C$.
    Since each is a basis for $\mathcal{H}$, observables in $B$ commute with observables in $C$ on their shared eigenbasis, so all observables in $B$ and $C$ pairwise commute, and thus $B\sim C$.
\end{proof}

Let $C_\mathcal{H}$ be the set of contexts over the Hilbert space $\mathcal{H}$.
Then since compatibility is an equivalence relation on $C_\mathcal{H}$, it induces a partition of $C_\mathcal{H}$ into compatible equivalence classes.
We define a \emph{supercontext} to be the union of all contexts in a compatible class.
Thus, a supercontext is a maximal set of commuting observables.
A supercontext is itself a context, and contains every context compatible with it: therefore, the supercontext in any compatible class is unique.

We now prove that the outcomes of all measurements in a context are determined by the outcomes of the measurements in any subset that is itself a context.
\begin{theorem}
	\label{contextdetermining}
	Given a context $\mathcal{A}$, if $\mathcal{B}=\{B_j\}_{j=1}^J\subseteq\mathcal{A}$ is also a context, then the outcome of any measurement $A\in\mathcal{A}$ is uniquely specified by the outcomes of the measurements $\mathcal{B}$.
\end{theorem}
\begin{proof}
	If $A\in\mathcal{B}$, then the result is trivially true.
	Now suppose $A\notin\mathcal{B}$.
	Let $\{|b^{(i)}\rangle\}$ be the common eigenstates of the observables $\{B_j\}=\mathcal{B}$.
	Since $\mathcal{B}$ is a context, the $\{|b^{(i)}\rangle\}$ are uniquely specified.
	Therefore they are eigenstates of any observable in $\mathcal{A}$, since $\mathcal{A}$ is itself a context: in particular they are eigenstates of $A$.
	Let $\{b_1^{(i)},...,b_J^{(i)}\}$ be the eigenvalues of $|b^{(i)}\rangle$ under $\{B_1,B_2,...,B_J\}$, and let $a^{(i)}$ be the eigenvalue of $|b^{(i)}\rangle$ under $A$.
	Then since each $|b^{(i)}\rangle$ is uniquely specified by $\{b_1^{(i)},...,b_J^{(i)}\}$, if on any initial state we perform measurements of all of the $\{B_j\}$ we will project onto one of the common eigenstates $|b^{(i)}\rangle$, so if we subsequently perform the measurement $A$ we will obtain the outcome $a^{(i)}$ with certainty.
	Thus there is a unique map $f$ defined by
	\begin{equation}
		f(b_1^{(i)},...,b_J^{(i)})=a^{(i)}\quad\forall i,
	\end{equation}
	giving the outcome for $A$ determined by any joint outcome for $\mathcal{B}$.
\end{proof}

In the main text, we defined $\overline{\mathcal{S}}$ to be the set of measurements determined by $\mathcal{S}$.
It is worth giving a simple example of this: if $\mathcal{S}=\{XI,IX\}$, then $\overline{\mathcal{S}}=\{XI,IX,XX\}$.
We can see that if $\mathcal{S}$ is a context (as is the case in this example), $\overline{\mathcal{S}}$ will be the unique supercontext that contains $\mathcal{S}$.
No assignment of outcomes to $\mathcal{S}$ is contradictory in this case, but the assignment $(1,1,-1)$ to $\overline{\mathcal{S}}$ is contradictory since it violates the operator relations among $XI$, $IX$, and $XX$ (namely, that any one is the product of the other two.)
The four consistent assignments to $\overline{\mathcal{S}}$ are thus $(1,1,1)$, $(1,-1,-1)$, $(-1,1,-1)$, and $(-1,-1,1)$: note that each of these corresponds to a unique assignment to $\mathcal{S}$.
A joint outcome for any set of (commuting) measurements actually performed on a quantum system will always be consistent in this way.

\section{Quantifying contextuality}
\label[appendix]{measures}

Given the methods developed in the main text for assessing contextuality as a true or false property of a set of measurements, we may extend this to a measure of the amount of contextuality.
As noted in in the main text, for VQE one natural choice for a contextuality measure is the distance (using any operator norm) of the given Hamiltonian from any noncontextual Hermitian operator:
\begin{equation}
    \label{csep}
    \text{CSep}(H)\equiv\min_{H'}\left(\frac{\parallel H-H'\parallel}{\parallel H\parallel}\right),
\end{equation}
where $H'$ is any noncontextual Hermitian operator, and $\parallel\cdot\parallel$ is some operator norm.
We call this measure the \emph{contextual separation}, or CSep.

The Pauli operators are a Hilbert-Schmidt orthogonal basis for the Hermitian operators.
Thus, any noncontextual set of Pauli operators defines a subspace of the Hermitian operators.
Any noncontextual Hamiltonian $H'$ is an element of one of these subspaces, so the minimum in \eqref{csep} is achieved by setting $H'$ to be the maximal projection of $H$ onto any noncontextual subspace.
Since $H$ in the form
\begin{equation}
    H=\sum_ih_i\mathcal{P}_i
\end{equation}
is written as a vector with components $h_i$ in the coordinate system defined by the Pauli operators $\mathcal{P}_i$, the projection of $H$ onto the subspace spanned by a set $\mathcal{T}$ of Pauli operators is
\begin{equation}
    H'=\sum_{\mathcal{P}_i\in\mathcal{T}}h_i\mathcal{P}_i.
\end{equation}
Therefore, if we take the norm in \eqref{csep} to be the Hilbert-Schmidt norm, then \eqref{csep} may be written
\begin{equation}
    \text{CSep}(H)=\min_{\mathcal{T}\subseteq\mathcal{S}}\left(\frac{\parallel\sum_{\mathcal{P}_i\notin\mathcal{T}}h_i\mathcal{P}_i\parallel}{\parallel\sum_ih_i\mathcal{P}_i\parallel}\right),
\end{equation}
where $\mathcal{S}=\{\mathcal{P}_i\}$ (as above) and $\mathcal{T}$ is any noncontextual subset of $\mathcal{S}$.
More conveniently,
\begin{equation}
    \label{csep2}
    \text{CSep}(H)=\min_{\mathcal{S}'\subseteq\mathcal{S}}\left(\frac{\parallel\sum_{\mathcal{P}_i\in\mathcal{S}'}h_i\mathcal{P}_i\parallel}{\parallel\sum_ih_i\mathcal{P}_i\parallel}\right)=\min_{\mathcal{S}'\subseteq\mathcal{S}}\sqrt{\frac{\sum_{\mathcal{P}_i\in\mathcal{S}'}h_i^2}{\sum_ih_i^2}},
\end{equation}
where $\mathcal{S}'$ is any decontextualizing set, as defined in the main text.

Let us define $\vec{h}\equiv(h_1,...,h_m)$, and let $\vec{h}'$ be the projection of $\vec{h}$ onto the span of any subset of the standard basis such that the support of $\vec{h}-\vec{h}'$ corresponds to a noncontextual subset of $\mathcal{S}$.
(In other words, the set of measurements corresponding to the support of $\vec{h}'$ is a decontextualizing set.)
Then \eqref{csep2} assumes the useful form
\begin{equation}
    \label{csep3}
    \text{CSep}(H)=\min_{\vec{h}'}\sqrt{\frac{\vec{h}'^2}{\vec{h}^2}}=\min_{\vec{h}'}\left(\frac{\parallel\vec{h}'\parallel_2}{\parallel\vec{h}\parallel_2}\right).
\end{equation}
This suggests another name for the contextual separation: the \emph{contextual 2-distance}, since \eqref{csep3} is the (scaled) 2-distance between the vector $\vec{h}$ and its maximal noncontextual projection.
We may then generalize to the \emph{contextual $p$-distance}:
\begin{equation}
    \label{cdp_appendix}
    \text{CD}_p(H)=\min_{\vec{h}'}\left(\frac{\parallel\vec{h}'\parallel_p}{\parallel\vec{h}\parallel_p}\right),
\end{equation}
of which contextual separation is the special case for $p=2$.

We thus have a family of measures of contextuality for any Hermitian operator.
The contextual 1-distance $\text{CD}_1(H)$ is the minimum absolute fractional weight of any decontextualizing set.
Thus (as noted in the main text) it has a physical interpretation as an upper bound on the fractional error induced in the energy estimate by ``decontextualizing" the procedure.
As noted above, $\text{CD}_2(H)$ is the minimum Hilbert-Schmidt distance of the Hamiltonian from a noncontextual Hamiltonian.
For the contextual $p$-distances for $p>2$ we do not have such simple physical interpretations, although $\text{CD}_\infty(H)$ is the minimum over all decontextualizing sets $\mathcal{T}$ of the maximum $h_i$ associated to $\mathcal{T}$, as a fraction of the maximum $h_i$ over the entire Hamiltonian.

Prior measures of contextuality include the contextual fraction (CF) (in \cite{abramsky11a,abramsky17a,mansfield18a,duarte18a}), relative entropy of contextuality (REC), mutual information of contextuality (MIC), and contextual cost (CC) (all in \cite{grudka14a}), and rank of contextuality (RC) \cite{horodecki18a}.
In a strict sense our contextual $p$-distance is a complementary measure to CF and CC, both of which measure the fraction of an empirical model that must be strongly contextual.
In particular, Proposition 6.3 in \cite{abramsky11a} states that $\text{CF}=1$ if and only if the model is strongly contextual, which means that $\text{CD}_p>0$ if and only if $\text{CF}=1$ (and correspondingly, $\text{CC}=1$).
Thus $\text{CD}_p$ and $\text{CF}/\text{CC}$ vary over disjoint regions in the space of empirical models.
MIC and REC are shown to be equal in \cite{grudka14a}, and are related to contextuality as a resource for communication.

Rank of contextuality is the measure most closely related to our $\text{CD}_p$, being the minimum number of noncontextual empirical models (``boxes" in the terminology of \cite{horodecki18a}) required to simulate the system of interest.
However, rank of contextuality and $\text{CD}_p$ are not even necesarily monotonically related, since an adversary could construct a Hamiltonian for which many noncontextual boxes are required for an exact description, but for which the weights of all but one of these boxes are arbitrarily small, thus giving a low $\text{CD}_p$.
It is possible that one could define a weighted version of the rank of contextuality that would avoid this problem, or that rank of contextuality might have some other operational meaning in the variational quantum eigensolver, but we do not pursue this herein.

Calculating $CD_p$ via compatibility graphs involves an optimization problem over subgraphs that are disjoint unions of cliques, by \cref{commutationthm}.
Thus evaluating $CD_p$ by strictly graph-theoretic methods is a variant of the clique problem, and is therefore likely to be NP-complete, so any efficient method for evaluating the contextual $p$-distance will have to take advantage of the structure of commutation relations that goes beyond compatibility graphs.
Finding such a method is an open question.

\label[appendix]{relations}

In \cite{abramsky11a,ramanathan12a}, the authors do not require that the set of measurements forms a closed subtheory.
As a result, our criterion in \cref{contextualcondition} is different from Proposition 1 in \cite{ramanathan12a} (originally proven in a non-quantum setting in \cite{vorobyev63a,vorobyev67a}), which states that a set of measurements admits a joint probability distribution if its compatibility graph is chordal~\footnote{A chordal graph contains no induced cycles with length greater than 3.}.
Note that the first two graphs in \cref{4forms_app} are chordal, and would thus be classified as noncontextual by Proposition 1 in \cite{ramanathan12a}.

The criterion of \cite{ramanathan12a} is valuable in capturing strong contextuality as it may exist strictly internally in a set of measurements.
In \cite{karanjai18a} it is demonstrated that for a quantum procedure the efficiency of classical simulation is limited by the presence of contextuality, as noted above: efficient simulation by sampling from the discrete Wigner function is only possible in the absence of contextuality.
In showing this, as noted in the introduction the authors assume that their sets of measurements are closed subtheories \cite[pp. 1-2]{karanjai18a}, which means exactly that all elements of $\overline{\mathcal{S}}$ must be included.
Thus the condition for contextuality we have developed is that upon which their argument is based.

In addition to providing a connection to the simulability results in \cite{karanjai18a}, requiring that the set of measurements be a closed subtheory is important when interpreting a noncontextual joint probability distribution as an ontological hidden-variable theory. If a joint outcome assignment to operators including a commuting pair $A,B$ does not imply the corresponding assignment to $AB$, it is difficult to interpret the original assignment as an ontic state of the system. Indeed, such a state is manifestly contextual in the sense that the ontic values can only apply if certain measurements are disallowed. Impossibility of a local-realistic hidden-variable theory for a set of measurements is commonly regarded as equivalent to contextuality of that set (see the introductory discussions in \cite{ramanathan12a,howard14a,grudka14a,cabello15a,cabello18a,raussendorf18a}, for example), but apparently we must be careful in associating the two.

The converse of Proposition 1 of \cite{ramanathan12a} is proven in \cite{cabello18a}: this implies that since the third compatibility graph in \cref{4forms_app} (the 4-cycle) is nonchordal, unlike the other two it is contextual both by our condition and by that of \cite{cabello18a}.
However, \cite{cabello18a} uses the fact that a 4-cycle compatibility graph is equivalent to the CHSH scenario \cite{clauser69a,fine82a,araujo13a}, and so the contradiction for this case is derived from violation of an inequality rather than directly from outcome assignability.

\section{VQE experiments to date}
\label[appendix]{vqedetails}

Small scale VQE experiments have already been performed in numerous systems \cite{peruzzo2014variational,omalley16a,shen2017quantum,kandala2017hardware,PhysRevX.8.031022,dumitrescu18a,nam19a,kandala19a}.
We used the methods developed in the main text to evaluate the contextuality of these: the results are given in Table I, in the main text.
Some other details of these experiments are given in \cref{vqedetailstable}; the first block also repeats the information in Table I, for reference.
We choose to present $\text{CD}_0$ because, unlike the other $\text{CD}_p$, it is independent of the coefficients, and many of the experiments given in this table use ranges of values for the coefficients.

As noted above, calculating $\text{CD}_p$ for any $p$ (including $p=0$) is in general hard, since it involves an optimization over subsets of the terms in the Hamiltonian.
For the contextual experiments we consider here, however, we may find $\text{CD}_0$ by brute force search for those with fewer terms (namely, \cite{peruzzo2014variational,PhysRevX.8.031022,nam19a}), and those with larger numbers of terms (namely, \cite{kokail19a,kandala2017hardware,kandala19a}), the compatibility graphs are sufficiently structured to enable a greedy heuristic approximation of $\text{CD}_0$.
In particular, all of these sets of terms contain commuting subsets that comprise substantial fractions of the full sets.
Therefore, we approximate the largest noncontextual subset by first including the largest commuting subset, then including the largest possible second clique (i.e., a commuting subset that anticommutes with some subset of the first commuting set), then including the largest possible third clique, and so forth.
This greatly restricts the number of noncontextual subsets we have to consider, and renders the optimization tractable.
We expect this heuristic to give a good approximation to the largest noncontextual subset when the approximate largest noncontextual subsets thus found are of comparable size to the full set of terms: this is the case for the Hamiltonians in \cite{kokail19a,kandala2017hardware,kandala19a}.
We also find that for the Hamiltonians in \cite{peruzzo2014variational,PhysRevX.8.031022,nam19a}, for which we obtained the exact largest noncontextual subsets, our heuristic approach also finds the exact solutions.

Our heuristic approach may still be inefficient if it is hard to find the largest commuting subset of the set of terms, but for the Hamiltonians in \cite{kokail19a,kandala2017hardware,kandala19a} this task turns out to be simple.
In \cite{kokail19a} for large $N$ nearly all of the terms are diagonal ($\frac{1}{2}N^2-\frac{1}{2}N+3$ terms are diagonal, while $2(N-1)$ are not).
In the LiH Hamiltonian in \cite{kandala2017hardware,kandala19a} the diagonal terms form a maximal commuting set (for $n$ qubits there can be no more than $2^n-1$ non-identity commuting Pauli operators). 
Finally, in the BeH Hamiltonian in \cite{kandala2017hardware} the diagonal terms form a maximal commuting set minus one element (and no other commuting subset is larger).

\setcounter{table}{1}
\begin{table*}[htp]
  \begin{tabular}{ l l l c c r }
    Citation: & System\quad\ & $|\mathcal{S}|$ & Contextual? & $CD_0$ & Encoding \\
    \hline
    Dumitrescu \emph{et al.}, 2018 \cite{dumitrescu18a}\quad\ & Deuteron\quad\ & various\footnote{In \cite{dumitrescu18a}, Dumitrescu \emph{et al.} compute the ground-state energies of the deuteron for effective field theories with dimension $2^N$ for $N=1,2,3$ ($|\mathcal{S}|=1,4,3$), and extrapolate from these to the infinite-dimensional space.
    Thus, $\mathcal{S}$ is different for each value of $N$.
    For $N=1$, $\mathcal{S}=\{Z\}$.
    For $N=2$, $\mathcal{S}=\{ZI,IZ,XX,YY\}$.
    For $N=3$, $\mathcal{S}=\{IZ,XX,YY\}$.
    All of these are noncontextual.} & No & 0 & JW \\
    \hline
    Kandala \emph{et al.}, 2017 \cite{kandala2017hardware}\quad\ & H$_2$\quad\ & 4 & No & 0 & hybrid \\
    \hline
    O'Malley \emph{et al.}, 2016 \cite{omalley16a}\quad\ & H$_2$\quad\ & 5 & No & 0 & BK \\
    \hline
    Hempel \emph{et al.}, 2018 \cite{PhysRevX.8.031022}\quad\ & H$_2$\quad\ & 5 & No & 0 & BK \\
    \hline
    Hempel \emph{et al.}, 2018 \cite{PhysRevX.8.031022}\quad\ & H$_2$\quad\ & 14 & No & 0 & JW \\
    \hline
    Colless \emph{et al.}, 2018 \cite{colless18a}\quad\ & H$_2$\quad\ & 5 & No & 0 & BK \\
    \hline
    Kokail \emph{et al.}, 2018 \cite{kokail19a}\quad\ & Lattice Schwinger Model\quad\ & 231 & Yes & $\sim$0.16 & JW \\
    \hline
    Nam \emph{et al.}, 2019 \cite{nam19a}\quad\ &$\text{H}_2$O\quad\ & 22 & Yes & 0.27 & JW \\
    \hline
    Hempel \emph{et al.}, 2018 \cite{PhysRevX.8.031022}\quad\ & LiH\quad\ & 13 & Yes & 0.33 & BK \\
    \hline
    Peruzzo \emph{et al.}, 2014 \cite{peruzzo2014variational}\quad\ & HeH$^+$\quad\ & 8\quad\ & Yes & 0.38 & JW \\
    \hline
    Kandala \emph{et al.}, 2017 \cite{kandala2017hardware}\quad\ & BeH\quad\ & 164 & Yes & $\sim$0.74 & hybrid \\
    \hline
    Kandala \emph{et al.}, 2017/19 \cite{kandala2017hardware,kandala19a}\quad\ & LiH\quad\ & 99 & Yes & $\sim$0.77 & hybrid
  \end{tabular}
$\vspace{0.2in}$
  \begin{tabular}{ l l l l }
    Citation: & System\quad\ & $\underline{\mathcal{S}}$ & Error\\
    \hline
    Dumitrescu \emph{et al.}, 2018 \cite{dumitrescu18a}\quad\ & Deuteron\quad\ & various & $<3\%$ \\
    \hline
    Kandala \emph{et al.}, 2017 \cite{kandala2017hardware}\quad\ & H$_2$ & $\{ZI,IZ,XX\}$ & $<1.6$ mHa\footnote{1.6 mHa is chemical accuracy.} \\
    \hline
    O'Malley \emph{et al.}, 2016 \cite{omalley16a}\quad\ & H$_2$ & $\{ZI,IZ,XX\}$ & $<1.3$ mHa \\
    \hline
    Hempel \emph{et al.}, 2018 \cite{PhysRevX.8.031022}\quad\ & H$_2$ (BK) & $\{ZI,IZ,XX,YY\}$ & $F>0.99$ \\
    \hline
    Hempel \emph{et al.}, 2018 \cite{PhysRevX.8.031022}\quad\ & H$_2$ (JW) & $\{Z_i|i=1,...,4\}\cup\{XXYY,YYXX,YXXY,XYYX\}$\quad\ & $F>0.97$ \\
    \hline
    Colless \emph{et al.}, 2018 \cite{colless18a}\quad\ & H$_2$& $\{ZI,IZ,XX,YY\}$ & various \\
    \hline
    Kokail \emph{et al.}, 2018 \cite{kokail19a}\quad\ & L.S.M.\quad\ & $\{X_iX_{i+1},Y_iY_{i+1},Z_i|i=1,2,...,19\}\cup\{Z_{20}\}$\footnote{For brevity, we represent some operators in this table as Pauli operators with subscripts: the subscript indicates which qubit the operator acts on. So for example, $X_2X_3$ means $IXX$ (in a three-qubit system.)} & $F>0.8$ \\
    \hline
    Nam \emph{et al.}, 2019 \cite{nam19a}\quad\ & $\text{H}_2$O & $\{Z_i|i=1,...,5\}\cup\{X_1X_2,Y_1Y_2,X_3X_4,X_3X_5,X_4X_5,Y_3Y_4,Y_3Y_5,Y_4Y_5\}$\quad\ & $<1.6$mHa \\
    \hline
    Hempel \emph{et al.}, 2018 \cite{PhysRevX.8.031022}\quad\ & LiH & $\{X_i,Z_i|i=1,2,3\}$ & various \\
    \hline
    Peruzzo \emph{et al.}, 2014 \cite{peruzzo2014variational}\quad\ & HeH$^+$ & $\{X_i,Z_i|i=1,2\}$ & 3.05mHa \\
    \hline
    Kandala \emph{et al.}, 2017 \cite{kandala2017hardware}\quad\ & BeH & $\{X_i,Z_i|i=1,...,6\}$ & $<1.6$mHa\\
    \hline
    Kandala \emph{et al.}, 2017/19 \cite{kandala2017hardware,kandala19a}\quad\ & LiH & $\{X_i,Z_i|i=1,...,4\}$ & $<1.6$mHa
  \end{tabular}
\caption{Evaluation of contextuality in VQE experiments}\label{vqedetailstable}
\end{table*}

\end{document}